\documentclass[11pt]{article}
\usepackage[T1]{fontenc}
\usepackage{tgpagella}

\usepackage{a4}
\usepackage{array}
\usepackage{tabularx}
\usepackage{graphicx}
\usepackage{amsmath,amssymb,amsthm,mathtools}
\usepackage{paralist}
\usepackage{bm}
\usepackage{bbm}
\usepackage{xspace}
\usepackage{url}
\usepackage{fullpage, prettyref}
\usepackage{boxedminipage}
\usepackage{wrapfig}
\usepackage{ifthen}
\usepackage{color}
\usepackage{xcolor}
\usepackage{framed}
\usepackage[pagebackref,letterpaper=true,colorlinks=true,pdfpagemode=none,urlcolor=blue,linkcolor=blue,citecolor=violet,pdfstartview=FitH]{hyperref}
\usepackage[nameinlink]{cleveref}
\usepackage{fullpage}
\usepackage{esvect}
\usepackage{thmtools}
\usepackage{thm-restate}

\newtheorem{theorem}{Theorem}[section]

\newtheorem{lemma}[theorem]{Lemma}
\newtheorem{claim}[theorem]{Claim}
\newtheorem{corollary}[theorem]{Corollary}
\newtheorem{definition}[theorem]{Definition}

\newtheorem{observation}[theorem]{Observation}
\newtheorem{remark}[theorem]{Remark}

\newcommand{\ignore}[1]{}

\newcommand{\Sec}[1]{\hyperref[sec:#1]{\Cref*{sec:#1}}} 
\newcommand{\Eqn}[1]{\hyperref[eq:#1]{(\ref*{eq:#1})}} 
\newcommand{\Fig}[1]{\hyperref[fig:#1]{Fig.\,\ref*{fig:#1}}} 
\newcommand{\Tab}[1]{\hyperref[tab:#1]{Tab.\,\ref*{tab:#1}}} 
\newcommand{\Thm}[1]{\hyperref[thm:#1]{Theorem\,\ref*{thm:#1}}} 
\newcommand{\Fact}[1]{\hyperref[fact:#1]{Fact\,\ref*{fact:#1}}} 
\newcommand{\Lem}[1]{\hyperref[lem:#1]{Lemma\,\ref*{lem:#1}}} 
\newcommand{\Prop}[1]{\hyperref[prop:#1]{Prop.~\ref*{prop:#1}}} 
\newcommand{\Cor}[1]{\hyperref[cor:#1]{Corollary~\ref*{cor:#1}}} 
\newcommand{\Conj}[1]{\hyperref[conj:#1]{Conjecture~\ref*{conj:#1}}} 
\newcommand{\Def}[1]{\hyperref[def:#1]{Definition~\ref*{def:#1}}} 
\newcommand{\Alg}[1]{\hyperref[alg:#1]{Alg.~\ref*{alg:#1}}} 
\newcommand{\Obs}[1]{\hyperref[obs:#1]{Obs.~\ref*{obs:#1}}} 
\newcommand{\Ex}[1]{\hyperref[ex:#1]{Ex.~\ref*{ex:#1}}} 
\newcommand{\Clm}[1]{\hyperref[clm:#1]{Claim~\ref*{clm:#1}}} 
\newcommand{\Step}[1]{\hyperref[step:#1]{Step~\ref*{step:#1}}} 

\usepackage[utf8]{inputenc} 
\usepackage[T1]{fontenc}    
\usepackage{hyperref}       
\usepackage{url}            
\usepackage{amsfonts}       

\usepackage{comment}

\usepackage[ruled,linesnumbered]{algorithm2e}
\usepackage{algorithmic}

\hypersetup{linktocpage = true, linkcolor = blue}

\newcommand{\R}{\mathbb{R}}
\newcommand{\NN}{\mathsf{NN}}
\newcommand{\vol}{\operatorname{vol}}

\newcommand{\Bopen}{B}
\newcommand{\aff}{\operatorname{aff}}

\newcommand{\proj}{\operatorname{proj}}
\newcommand{\Sph}{\mathbb{S}}

\DeclareMathOperator*{\argmin}{arg\,min}

\usepackage[framemethod=TikZ]{mdframed}

\newmdenv[
  roundcorner=10pt,
  linewidth=0pt,          
  backgroundcolor=orange!10,
  innertopmargin=2pt,     
  innerbottommargin=10pt, 
  innerleftmargin=10pt,
  innerrightmargin=10pt
]{thmbox}

\newmdenv[
  roundcorner=10pt,
  linewidth=0pt,          
  backgroundcolor=green!10,
  innertopmargin=2pt,     
  innerbottommargin=10pt, 
  innerleftmargin=10pt,
  innerrightmargin=10pt
]{lemmabox}

\newmdenv[
  roundcorner=10pt,
  linewidth=0pt,          
  backgroundcolor=blue!10,
  innertopmargin=2pt,     
  innerbottommargin=10pt, 
  innerleftmargin=10pt,
  innerrightmargin=10pt
]{defbox}

\newenvironment{proofof}[1]{{\bf Proof of #1.  }}{\hfill$\Box$}

\usetikzlibrary{decorations.markings}
\usepackage{floatrow}

\title{Learning Nearest-Neighbor Maps from Adaptive Queries\footnote{This work began while both authors were affiliated with University of California, San Diego and the EnCORE Institute, funded by NSF TRIPODS Institute grant 2217058.}}

\author{
Hadley Black\thanks{CUNY Baruch College. {\tt hadley.black@baruch.cuny.edu}.} \\  
\and 
Geelon So\thanks{ETH Z\"urich. {\tt geelon.so@inf.ethz.ch}.} \\
}
\date{August 8, 2026}
\begin{document}

\maketitle
\thispagestyle{empty}

\begin{abstract}
We study the problem of learning nearest-neighbor maps from adaptive queries, which is equivalent to the following problem of reconstructing a hidden set $H$ via a nearest-neighbor query oracle. Let $K \subset \mathbb{R}^d$ be a compact domain in a normed space $(\R^d,\| \cdot\|)$ and let $H \subset K$ be a hidden set of $n$ points. Upon querying $q \in K$, the oracle returns some $h \in H$ with minimum distance from $q$. How many queries are required to exactly recover $H$? Previous work has studied this question in specific domains, namely the Boolean hypercube and the $\ell_2$-unit sphere. We generalize previous work and prove the tight worst-case query complexity bound of $\Theta(n\kappa)$, where $\kappa$ is the kissing number of the underlying norm.

In the Euclidean norm, obtaining tight asymptotic bounds on $\kappa$ is a significant open question, although it is known that $\kappa = \exp(\Theta(d))$. Our second set of results shows that an exponential dependence on $d$ is required even in natural Euclidean domains: $\exp(\Omega(d))$ queries are needed in the ball, even when $n=2$, and $n\exp(\Omega(d))$ queries are needed in the cone. 

Lastly, we prove a sharper upper bound in the Euclidean sphere. Here, $d$ can be replaced by $\min(n,d)$ via a dimension reduction preprocessing step. This is a randomized version of a procedure due to Prabhu-Woodruff (ICML 2024) where we improve the query complexity from $O(nd)$ to $O(\min(n,d))$. This reveals a striking contrast between the sphere and the ball: when $n = O(1)$, the sphere admits an $O(1)$ query algorithm, whereas the ball requires $\exp(\Omega(d))$. 
\end{abstract}

\newpage

\clearpage
\pagenumbering{arabic}

\section{Introduction}

We study \emph{learning nearest-neighbor maps from adaptive queries}. Let $K$ be a compact metric space, and let $H \subseteq K$ be a finite set of hidden points. A map $\NN_H : K \to K$ is a nearest-neighbor map of $H$ if, for each $q \in K$, it returns a point in $H$ that minimizes distance to $q$. In this problem, the goal of a learner is to recover any nearest-neighbor map of $H$. To achieve this, the learner may make queries to a nearest-neighbor oracle: given $q \in K$, the oracle returns a nearest neighbor $\NN_H(q)$. Furthermore, the learner may adaptively interact with the oracle over many rounds. We would like to understand the \emph{query complexity} of learning a nearest neighbor map: how many queries does a learner need to compute any nearest-neighbor map of $H$, and how should it select them?

There are a few natural ways to view this problem. In a model of \emph{exact query learning} \cite{DBLP:journals/ml/Angluin87}, a learner would like to identify an underlying target function by adaptively probing it. The problem we study extends the query learning model to the class of nearest-neighbor maps, which is in itself fundamental to many areas of computer science, machine learning, information retrieval, and computational geometry \cite{DBLP:conf/soda/Yianilos93,devroye1996probabilistic,indyk1998approximate,10.1145/116873.116880}. Additionally, the nearest-neighbor map is often treated as a black-box primitive for solving downstream problems in computational geometry, such as clustering, estimating statistics about a point set, simulating a separation oracle, computing range queries, to name a few \cite{DBLP:journals/dcg/Har-PeledKMR16,DBLP:conf/kdd/DalviKMR11,DBLP:journals/geoinformatica/BaeAKNS09,DBLP:journals/tkde/YanGZHZW16}. This black box is precisely the oracle that is available to our learner, and the goal here is to understand how many calls to the black box are needed to, as it were, open it up. 

Finally, there is a particularly elegant and equivalent formulation of the problem, which we will heavily make use of. It follows from the fact that a learner is able to compute a nearest-neighbor map of $H$ if and only if it knows what $H$ is. Thus, the problem is one of \emph{reconstructing the hidden point set $H$ from adaptive nearest-neighbor queries}. Two specific instances of the problem formulated this way were introduced and studied by \cite{PW24}. They focused on when the domain $K$ is the Boolean hypercube with Hamming distance, and when it is the unit sphere with Euclidean distance; both in high dimensions. The motivation comes from viewing the hidden set $H$ as a dictionary of binary or spherical codes, and the nearest-neighbor map as an error-correcting mechanism.

\subsection{Main Results}

In this work, we let the domain $K$ be any compact subset of a $d$-dimensional normed space $(\R^d, \|\cdot\|)$. For any hidden set $H \subseteq K$, the nearest-neighbor map $\NN_H$ of $H$ is defined by
\[\NN_H(q) \in \argmin_{h \in H}\, \|q - h\|\text{.}\]
This setting is reasonably general, and it includes the Boolean hypercube in $(\R^d,\ell_1)$ and the unit sphere in $(\R^d,\ell_2)$ studied in \cite{PW24}. We show that the complexity of any nearest neighbor map depends on fundamental geometric properties of the domain $K$ and the norm $\|\cdot\|$. Our core result characterizes the minimax query complexity of learning nearest-neighbor maps in normed spaces.

\paragraph{A minimax-optimal learner.} It turns out that the \emph{kissing number} controls the worst-case query complexity. The usual definition of this number is the maximum number of unit balls with disjoint interior that can be arranged to simultaneously touch (kiss) a central unit ball. In normed spaces, the kissing number is also called the Hadwiger number \cite{Swanepoel2018}, and has an equivalent definition:

\begin{definition}[Kissing number] \label{def:kissing}
    The \emph{kissing number} of the norm $\|\cdot\|$ is the maximum number of points that can be packed onto the unit sphere with pairwise distance at least one
    \[
        \kappa := \max\big\{ |U|:
        \|u\| = 1
        \ \text{and}\
        \|u-v\| \geq 1
        \text{ for all } u \neq v\in U \big\} \text{.}
    \]
\end{definition}

Intuitively, the kissing number quantifies how many `directions' there are at a fixed scale, and an elementary fact is that it is always upper bounded by $3^d$ \cite{hadwiger1957treffanzahlen}. If $|H|= n$, this work shows that the minimax query complexity of learning a nearest-neighbor map is $n \kappa$, up to factors of 2.

\begin{theorem}[A minimax-optimal learner] \label{thm:characterization}
    Let $\|\cdot\|$ have kissing number $\kappa$. For each $n \geq 2$:
    \begin{itemize}
        \item There is an adaptive, deterministic algorithm which, for any non-empty, compact domain $K \subset \R^d$, learns any unknown set $H \subseteq K$ of size $n$ using at most $2 n \kappa$ nearest-neighbor queries. Moreover, the algorithm does not need to know $n$ beforehand and terminates within $2n$ rounds.
        \item There is a compact domain $K \subset \R^d$ on which any randomized adaptive algorithm that exactly learns any $n$-point hidden set $H \subseteq K$ must make at least $\lfloor n/2 \rfloor \cdot (\kappa - 1)$ queries.
    \end{itemize}
\end{theorem}

Interestingly, a closely-related upper bound was previously discovered in passing; an algorithm designed to solve a clustering problem also happens to solve this one \cite[Observation 3.3]{DBLP:journals/dcg/Har-PeledKMR16}. It uses at most $c_d n$ queries where $c_d$ is an unspecified, dimension-dependent constant. \Cref{thm:characterization} pins down this constant, establishing a tight, and somewhat surprising, connection to the kissing number of the norm. It is arguably not obvious that the kissing number, which deals with balls, should be the right quantity for controlling hardness over arbitrary compact sets.  

The algorithm realizing the upper bound (\Alg{Voronoi}) is a simple greedy algorithm, it is completely general, and this result shows that it is essentially optimal against worst-case problem instances. A natural question to then ask is: what happens in more structured settings?

\paragraph{Upper bounds in Euclidean space.} Here, we focus on $d$-dimensional Euclidean space $(\R^d, \ell_2)$. The bounds in \Cref{thm:characterization} depend on the $\ell_2$-kissing number in $\R^d$, which we denote $\kappa_{2,d}$. 

Determining the value of $\kappa_{2,d}$ is a classic geometric problem going back to an argument in 1694 between Isaac Newton and David Gregory \cite{BoyvalenkovDodunekovMusin2012}. For interest, \Cref{table:kissing} gives a few examples where $\kappa_{2,d}$ has been computed. In general, the tightest known asymptotic bounds are
\[2^{(0.2075 + o(1)) d} \leq \kappa_{2,d} \leq 2^{(0.401 + o(1))d},\]
where the upper bound is given by \cite{KabatianskyLevenshtein1978} and the lower bound is by \cite{wyner1965capabilities}.

With these bounds on $\kappa_{2,d}$, it turns out to be quite interesting to compare \Cref{thm:characterization} with the earlier results in \cite{PW24}. In the following, we consider learning an $n$-point hidden set $H$ in $K$.
\begin{itemize}
    \item Let $K$ be the unit $\ell_2$-sphere. Our algorithm uses $O(n2^{0.402d})$ queries, beating the prior query complexity of $O(n^{\lfloor d/2\rfloor})$ by a constant factor in the exponent (Theorem 4.2 of \cite{PW24}).
    \item Let $K$ be the Boolean hypercube $\{0,1\}^d$ embedded in $\ell_2$-space $(\R^d, \ell_2)$.\footnote{Technically, \cite{PW24} considered the Boolean hypercube with the $\ell_1$-distance. However, in the context of computing nearest neighbors, the $\ell_1$ and $\ell_2$ distances are equivalent when restricted to this domain. In particular, the distances are monotonically related since $\|x - y\|_1 = \|x - y\|_2^2$ whenever $x, y \in \{0,1\}^d$.} There is an algorithm specialized to the hypercube that uses $2^{O(\sqrt{d \cdot \log d \cdot \log n})}$ queries (Theorem 3.1 of \cite{PW24}). Thus, it uses significantly fewer queries than our \emph{worst-case} lower bound of $n \cdot 2^{\Omega(d)}$.
\end{itemize}
The second point does not contradict our lower bound as, evidently, the Boolean hypercube does not yield worst-case problem instances. It does point out that the kissing number only controls the \emph{minimax} query complexity, and that certain domains $K$ may exhibit additional structure.

This raises the question: how pathological is the lower bound in \Cref{thm:characterization}? Or, in other words, is learning the nearest-neighbor map easier on more natural domains?

\begin{table}[t] 
\begin{tabular}{lll}
$d$ & $\kappa_{2,d}$ & Reference \\ \hline
1   & 2              &  \emph{trivial}         \\
2   & 6              &  \emph{folklore}         \\
3   & 12             &  \cite{kissing3} \\
4   & 24             &  \cite{kissing4} \\
8   & 240            &   \cite{Lev79, ODLYZKO1979210}        \\
24  & 196560         &  \cite{Lev79, ODLYZKO1979210}  
\caption{\small{Values of $d$ where $\kappa_{2,d}$ has been determined. 
}
\label{table:kissing}
}
\end{tabular}
\end{table}

\paragraph{Lower bounds in Euclidean space.} The next result partially addresses the previous questions. We provide quantitative lower bounds for specific domains---namely the closed $\ell_2$-ball and cone. They arguably show that bounds obtained by the kissing number are fairly tight even in nice domains.

\begin{theorem}[Lower bounds in nice domains] \label{thm:lower-bound-nice-domains}
    Let $n \geq 2$. Suppose that a randomized adaptive algorithm exactly recovers any $n$-point hidden set $H \subset K$ with probability at least $1/10$.
    \begin{itemize}
        \item If $K \subset \R^d$ is the closed $\ell_2$-unit ball, the algorithm makes at least $\Omega\left(n + 2^{0.2075d}\right)$ queries.
        \item If $K \subset \R^d$ is the $\ell_2$-cone $C_d$, then the algorithm makes at least $\Omega\left(n \cdot 2^{0.0156d}\right)$ queries, where
        \[C_d := \big\{(s,v) \in [0,1] \times \R^{d-1} : \|v\|_2 \leq s\big\}.\]
    \end{itemize}
\end{theorem}

The lower bound for the ball is an additive bound of $n + 2^{\Omega(d)}$, where the constant in the exponent is $\frac{1}{2}\log_2(4/3) \approx 0.2075$. This corresponds to the tightest known lower bound of $\kappa_{2,d}$ \cite{wyner1965capabilities}. The lower bound for the cone gives a multiplicative bound $\smash{n \cdot 2^{\Omega(d)}}$. While we did not attempt to optimize the constant in this exponent, this still shows a lower bound in a `nice' domain that is on the same order as the worst-case lower bound $\Omega(n \cdot \kappa_{2,d})$ from \Cref{thm:characterization}. 

This lower bound gives some evidence that our minimax result may be indicative of hardness even for many non-worst-case problem instances. If so, this is unfortunate because a query complexity that is linear in $\kappa$, which is exponential in dimension for Euclidean space, quickly becomes intractable when $d$ grows. This leads us to our next question: are there reasonable settings where a learner, given additional structure, may circumvent the general bound $\Theta(n\cdot \kappa)$?

\paragraph{Learning a rank-deficient set on the sphere.}
For our final result, we let $K = \mathbb{S}^{d-1}$ be the $\ell_2$-unit sphere in $\R^d$. We also assume that the hidden set $H \subset K$ has rank $k = \dim(\mathrm{aff}(H))$, which is the dimension of the affine span of $H$. We do not assume that the learner knows the value of $k$.

This setting is inspired by \cite{PW24}, who observed that when the hidden set $H$ lies on the sphere, it is possible to exactly recover its affine span using $O(kd)$ deterministic nearest-neighbor queries. This subroutine enables a learner to efficiently reduce the dimensionality of the problem to $\min(k,d)$. We contribute a simplified, randomized dimension-reduction subroutine using $O(k)$ nearest-neighbor queries that almost surely recovers the affine span $\mathrm{aff}(H)$. 

\begin{lemma} [Randomized dimension reduction in $\Sph^{d-1}$] \label{lem:dim-reduction} 
    Let $H \subset \Sph^{d-1}$ be a set in the 
    $\ell_2$-unit sphere and with rank $k$. There is a randomized algorithm using $2k+3$ queries that almost surely returns an affine basis $u_1,\ldots,u_{k+1} \in H$. That is, with probability 1, their affine span contains $H$,
    \[H \subseteq \mathrm{aff}(u_1,\ldots, u_{k+1}).\]
    Moreover, the algorithm does not need to know $k$ beforehand and terminates within $k+2$ rounds.
\end{lemma}

By first running the subroutine provided by \Cref{lem:dim-reduction}, a learner can first use $2n + 3$ queries to recover $\mathrm{aff}(H)$, reducing the dimensionality of the problem to $k$. Then, by proceeding with the general algorithm \Alg{Voronoi} from \Cref{thm:characterization} within $\mathrm{aff}(H)$, using $2 n \cdot \kappa_{2,k}$ queries. We obtain: 

\begin{corollary} [Specialized algorithm in $\Sph^{d-1}$] \label{cor:sphere-alg} There is a randomized, adaptive algorithm that almost surely learns any $n$-point hidden set $H \subset \Sph^{d-1}$ using $(2n + 3) + 2n \cdot \min\{\kappa_{2,n}, \kappa_{2,d}\}$ nearest-neighbor queries. Moreover, the algorithm does not need to know $n$ beforehand and terminates within $3n + 1$ rounds.
\end{corollary}

Notice that this result shows a remarkable separation between the $\ell_2$-sphere and the $\ell_2$-ball. When there is only a constant number of hidden points, \Cref{cor:sphere-alg} yields a constant $O(1)$-query complexity on the sphere. In contrast, \Cref{thm:lower-bound-nice-domains} proves an exponential $\exp(\Omega(d))$ query complexity lower bound for the ball even when there are only two hidden points.

\subsection{Open Questions}

This work builds substantially on the prior work \cite{PW24}. Nevertheless, to our knowledge, it is also the first to study, in a fairly general setting, the problem of learning a nearest-neighbor map from adaptive queries. This simple question arises in the intersection of query learning and computational geometry. For normed spaces, our minimax query complexity already finds new connections to a classic notion of kissing numbers. Given how unexplored this problem is, there are many natural directions for follow-up. We list a few here.

\begin{itemize}
    \item \emph{Generalizations.} We studied the problem in normed spaces, but it also makes sense in general metric spaces, or other canonical families of metric spaces like graphs. What happens there?
    \item \emph{Average-case hardness and other hardness parameters.} Are there natural distributions under which random instances are easier? Furthermore, can we explain more generally why the problem is easier on the Boolean hypercube, and low-rank hidden sets on the sphere?
    \item \emph{Noisy and inexact oracles.} What happens if the nearest-neighbor oracle is not always correct or only outputs an approximate nearest neighbor?
    \item \emph{Fine-grained lower bounds.} Can our lower bounds for the ball and cone in $\ell_2$ be extended to all compact domains with non-empty interior? Can they be generalized to $\ell_p$-norms?
    \item \emph{Computation.} In this paper, we are concerned only with query complexity. Are there settings that admit computationally-efficient algorithms?
\end{itemize}

\section{Characterizing Query Complexity in Normed Spaces}

In this section we prove \Cref{thm:characterization} which characterizes the query complexity for a general compact set $K \subset \R^d$ in a finite-dimension real normed space, $(\mathbb{R}^d, \|\cdot\|)$, in terms of the kissing number, $\kappa := \kappa_{\|\cdot\|}$, of the space (\Cref{def:kissing}). In \Cref{sec:kissing-UB}, we show that there is a deterministic $2n$-round algorithm that recovers any hidden set $H \subseteq K$ of size $n$ using at most $2n \kappa$ queries. In \Cref{sec:LB-generic} we show that there always exists a compact set $K$ where $\Omega(n\kappa)$ queries are needed. 


\subsection{Upper Bound} \label{sec:kissing-UB}

Our algorithm works by maintaining a set of discovered points $X \subseteq H$, the first of which can be obtained with a single arbitrary query anywhere in $K$. Given non-empty $X$, we consider the decomposition of $K$ into Voronoi cells induced by $X$. For each $x \in X$, define
\begin{align} \label{eq:Voronoi}
    V_X(x) = \left\{z \in K \colon \|z-x\| \leq \|z-x'\| \text{ for all } x' \in X  \right\} 
\end{align}
as the set of points in $K$ whose nearest-neighbor in $H$ is $x$. Note that these cells cover $K$. A key observation is the following.

\begin{observation} \label{obs:alg} Querying any $q \in V_X(x)$ must either (i) discover a new $h \in H \setminus X$, or (ii) return $x$ (or an equidistant point $y \in X$ in the case that $q$ lies on the boundary of multiple cells). If (ii) occurs, we learn that no hidden points exist in the open ball around $q$ of radius $\|q-x\|$, i.e. 
\[
H \cap \Bopen(q,\|q-x\|) = \emptyset .
\]
\end{observation}  

With this in mind, suppose we query a set of points $Q \subseteq V_X(x) \setminus \{x\}$ such that these open balls form a \emph{cover} of $V_X(x)$:
\[
V_X(x) \setminus \{x\} \subseteq \bigcup_{q \in Q} \Bopen(q,\|q-x\|) \text{.}
\]
Then, by \Cref{obs:alg} after querying every point in $Q$ we are guaranteed to either (a) discover a new hidden point, or (b) certify that no more hidden points exist in $V_X(x)$, that is $H \cap V_X(x) = \{x\}$. In particular, once (b) occurs, we do not need to perform this procedure for $x$ going forward. In conclusion, we can perform this Voronoi cell covering procedure at most twice for every point $x \in H$: once to discover $x$, and once to certify the complete discovery of $H \cap V_X(x)$. As a result, the query complexity of the entire algorithm is bounded by $2n t$ where $t$ is an upper bound on the covering size. We prove the following lemma (proof deferred to \Cref{sec:tangent-cover}) which shows we can always find a covering of size at most $\kappa$, the kissing number of the norm, $\|\cdot\|$.

\begin{lemma}[Tangent-ball covering] \label{lem:tangent-cover}
Let $V\subseteq\R^d$ be a non-empty compact set in a normed space $(\R^d,\|\cdot\|)$ with kissing number $\kappa$ (\Cref{def:kissing}), and let $x\in V$.  There exists a set of points $Q \subseteq V\setminus\{x\}$, of size $|Q|\le \kappa$, such that
\begin{align} \label{eq:tangent-cover}
    V\setminus\{x\}
    \subseteq
    \bigcup_{q\in Q}
    \Bopen\bigl(q,\|q-x\|\bigr).
\end{align}
\end{lemma}

Given \Cref{lem:tangent-cover} and the above discussion, we state our algorithm formally in \Alg{Voronoi}. (See \Fig{Voronoi} for an illustration of an application of \Cref{lem:tangent-cover} during the algorithm's execution.)

\begin{algorithm}[ht]
\caption{Voronoi cell-covering algorithm\label{alg:Voronoi}} 
\textbf{Input:} A non-empty compact domain $K \subset \mathbb{R}^d$ in a normed space $(\R^d,\|\cdot\|)$ and nearest-neighbor query access to a hidden set $H \subseteq K$ of size $|H| = n$\; 
\textbf{Output:} The set $H$\;
\textbf{Initialization:} Let $x_0 \in H$ be the point returned by an arbitrarily chosen initial query in $K$. The set $X \gets \{x_0\}$ records the so-far discovered points and $C \gets \emptyset$ contains all $x \in X$ for which it has been certified that $H \cap V_X(x) \subseteq X$\;
\While{$\exists x \in X \setminus C$} { 
Let $V_X(x)$ denote the Voronoi cell of $x$ induced by $X$. (Recall \Cref{eq:Voronoi}.)\;
Apply \Cref{lem:tangent-cover} to construct a set $Q_x \subseteq V_X(x) \setminus \{x\}$ satisfying \Cref{eq:tangent-cover}\;
Query all points in $Q_x$. If some such query returns $h \in H \setminus X$, then set $X \gets X \cup \{h\}$. Otherwise, set $C \gets C \cup \{x\}$\;
}
\textbf{Return} $X$\;
\end{algorithm}

\begin{figure}
    \includegraphics[scale=0.4]{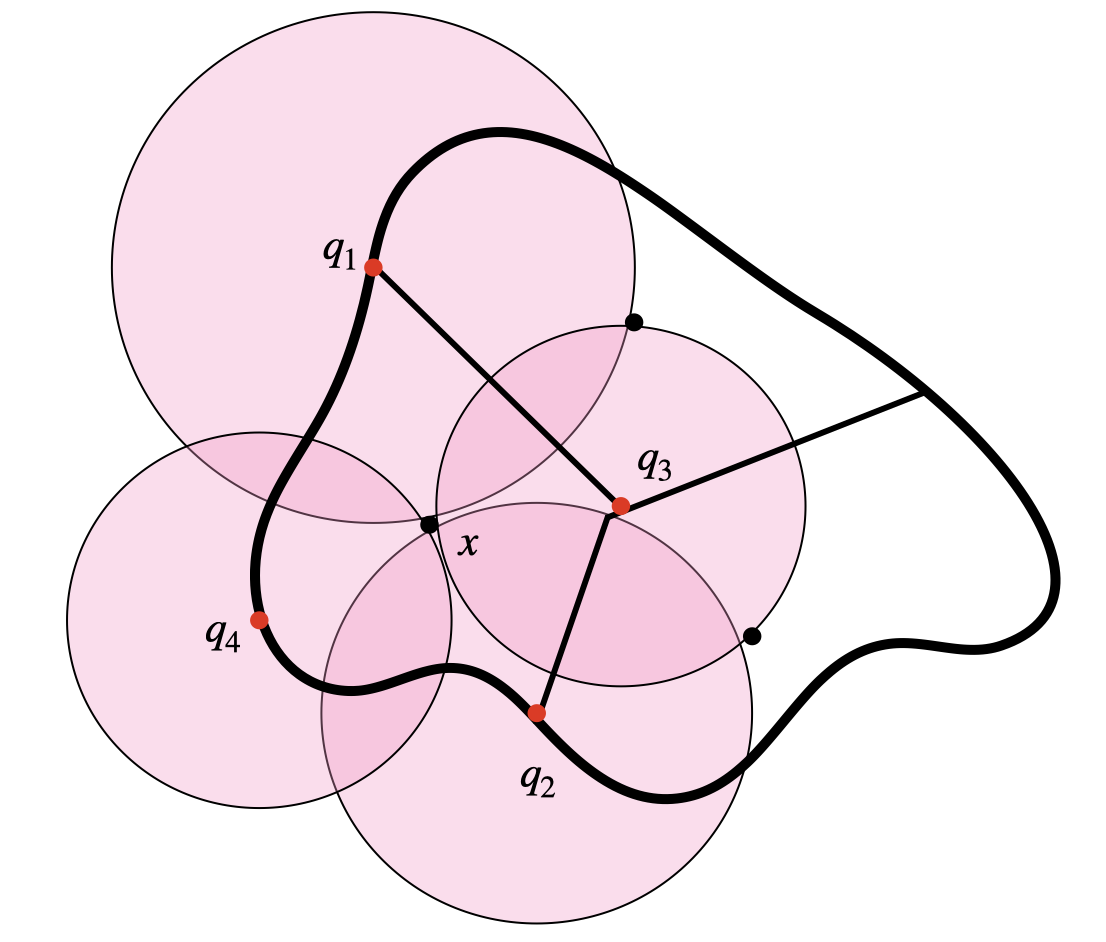}
    \caption{\small{An illustration showing an application of \Cref{lem:tangent-cover} in \Alg{Voronoi} under the $\ell_2$ norm. Black points denote the set $X$. Four queries $q_1,q_2,q_3,q_4 \in V_X(x)$ are selected such that the open balls $B(q_i,\|q_i-x\|)$ form a cover of the Voronoi cell, $V_X(x) \setminus \{x\}$.}}
    \label{fig:Voronoi}
\end{figure}

The correctness of \Alg{Voronoi} follows by the above discussion. That is, at termination the algorithm always returns the set $H$. To bound the query complexity, \Cref{lem:tangent-cover} shows that every round of queries made during an iteration of the while-loop either increments $|X|$ or increments $|C|$, using at most $\kappa$ queries. Each of these can be incremented at most $n$ times and so the total number of queries is at most $2n \kappa$. This completes the proof of the upper bound stated in \Cref{thm:characterization}.

\subsubsection{Proof of \Cref{lem:tangent-cover}} \label{sec:tangent-cover}

We construct $Q$ greedily while keeping track of the region of $V$ that is left uncovered. Initialize $R_0 = V$. We proceed in stages, selecting one point at a time, until $V$ has been covered. For $j \geq 1$, we let $q_j \in V$ denote the point chosen in stage $j$ and let
\[
    R_{j} = V \setminus \bigcup_{i \leq j}
    \Bopen\bigl(q_i,\|q_i-x\|\bigr) \text{.}
\]
Because the balls are open, $R_{j}$ is a closed subset of the compact set $V$, and is therefore compact. If $R_{j} = \{x\}$, we stop, since this means $V \setminus \{x\}$ has been covered. Otherwise, we select $q_{j+1} \in R_j$ to be a point in $R_j$ with maximum distance $\|q_{j+1}-x\|$ from $x$. Such a maximizer exists by compactness of $R_j$, and its distance from $x$ is non-zero since $R_j \neq \{x\}$. For every selected point $q_j$, write
\begin{align} \label{eq:radii}
    r_j=\|q_j-x\| >0 ~\text{ and }~
    u_j=\frac{q_j-x}{r_j}\in \mathbb{S}^{d-1}.
\end{align}
Since $R_0 \supset R_1 \supset \cdots \supset R_j$ are nested and $q_j$ maximizes the distance from $x$ in $R_{j}$, we have $r_1 \ge r_2 \ge \cdots \geq r_j$. Now, fix $i < j$.  At the time $q_j$ is selected, it has not been covered by the ball centered at $q_i$ and so $\|q_j-q_i\| \ge r_i$. Writing $q_i=x+r_i u_i$, $q_j=x+r_j u_j$, applying the triangle inequality and homogeneity, we obtain
\begin{align} \label{eq:code-bound1}
    r_i &\le \|q_i-q_j\| \nonumber \\
    &= \|r_i u_i-r_j u_j\| \nonumber \\
    &= \|(r_i - r_j)u_i + r_j(u_i-u_j)\| \nonumber \\ 
    &\le (r_i-r_j)\|u_i\| + r_j\|u_i - u_j\| = r_i - r_j +  r_j\|u_i - u_j\| \text{.}
\end{align}
Rearranging terms yields $\|u_i - u_j\| \geq 1$. Thus, the selected directions $U = \{u_i \colon i \geq 1\}$ satisfy $\|u_i - u_j\| \geq 1$ for all $i \neq j$. Hence, we have $|U| \leq \kappa$ by \Cref{def:kissing}. \\

\subsection{Lower Bound} \label{sec:LB-generic}

We now prove the following lower bound, showing that our algorithm in the previous section is optimal for an arbitrary compact $K$ in any finite-dimensional normed space.  

\begin{theorem}[Worst-case lower bound] \label{thm:worst-case-lower-bound}
    Suppose that a finite-dimensional normed space $(\R^d,\|\cdot\|)$ has kissing number $\kappa \geq 2$. There exists a compact set $K \subset \R^d$ such that any randomized adaptive algorithm that exactly recovers every $n$-point hidden set $H \subset K$ with probability $\Omega(1)$ must make at least $\Omega(n \kappa)$ queries.
\end{theorem}

\begin{proof}
    We first give a construction for the case $n = 2$. Let $\overline{B} \subset \R^d$ be a closed unit ball centered at the origin and let $U \subset \partial B$ be a set of points forming a $1$-packing of the sphere $\partial B$ of maximal size. In particular, by \Cref{def:kissing}, the size of $U$ is 
    \[|U| = \kappa.\]
    Fix the finite (compact) domain $K = U \cup \{0\}$. We let $H \subset K$ be a randomized hidden set consisting of the origin and a point $x \in U$ chosen uniformly at random
    \[H = \{0,x\}.\]
    Since $\|q - 0\| \leq 1 \leq \|q-x\|$ for all $q \in K \setminus \{x\}$, 
    a worst-case nearest neighbor oracle is
    \[\NN_H(q) = \begin{cases}
        0 & q \ne x\\
        x & q = x.
    \end{cases}\]
    In particular, any query $q \in U \setminus \{x\}$ reveals no information about the hidden point $x$ and so recovering $x$ is equivalent to an unstructured search on $|U| = \kappa$ elements. Therefore, to successfully recover $H$ with $\Omega(1)$ success probability requires $\Omega(\kappa)$ queries for any randomized algorithm. Specifically, achieving success probability $1$ requires at least $\kappa-1$ queries.

    This construction easily generalizes to a set of $n$ points, where we generate $k = \lfloor n/2 \rfloor$ independent, translated copies of the 2-point construction whose centers are sufficiently far apart. Let $v_1,\ldots, v_k \in \R^d$ be any $k$ points such that $\|v_i - v_j\| \geq 5$ for all $i \neq j$. 
    Then, let $x_1,\ldots, x_k \in U$ be selected independently and uniformly at random. Define $K_n$ and $H_n$ to be the sets
    \[K_n = \bigcup_{i=1}^k (v_i + K) \qquad \textrm{and}\qquad H_n = \bigcup_{i=1}^k \{v_i, v_i + x_i\}.\]
    By construction, every query $q \in v_i + K$ will return a nearest neighbor in $v_i + K$, so no queries can be shared across copies of the problem. Therefore, recovering $x_1,\ldots,x_k$ requires solving $k$ independent unstructured search problems, each on $\kappa$ elements. Therefore, $\Omega(n \kappa)$ queries are required for any randomized algorithm to achieve $\Omega(1)$ success probability. In the case where the success probability is 1, $\lfloor n/2\rfloor \cdot (\kappa-1)$ queries are required.
\end{proof}

\section{Lower Bounds in Euclidean Convex Bodies}

Our \Cref{thm:worst-case-lower-bound} constructs a worst-case domain on which the kissing number characterizes the query complexity. It is then natural to ask whether the problem is not as hard within a less pathological space. In this section, we partially address this question by proving lower bounds for natural Euclidean convex bodies, namely the ball and the cone.

\subsection{Additive Lower Bound in the Ball} \label{sec:LB-ball}


Our first result gives a hard construction when the domain is the closed unit ball in $\mathbb{R}^d$ under the $\ell_2$-distance. It shows that $\Omega(2^{0.2075d})$ nearest neighbor queries are required even in the case of a two-point hidden set---one point is the origin and the other point is chosen uniformly at random from the boundary of the ball. The additive lower bound $\Omega(n + 2^{0.2075d})$ stated in \Cref{thm:lower-bound-nice-domains} follows by combining this result with the trivial $\Omega(n)$ information-theoretic lower bound.

\begin{theorem}[Lower bound in the ball] \label{thm:lower-bound-ball}
    Let $K \subset \mathbb{R}^d$ be the closed unit $\ell_2$-ball. Every randomized adaptive algorithm that recovers every $2$-point hidden set with probability at least $\Omega(1)$ makes at least 
    \[\Omega\left((4/3)^{d/2}\right)\ \text{queries}.\]
\end{theorem}

\begin{proof}
    The following forms a hard, randomized 2-point hidden set, $H$. 
    Let $H \subset K$ be a hidden set consisting of the origin and a uniformly random point $x \in \mathbb{S}^{d-1}$ on the boundary of the unit ball,
    \[H = \{0, x\}.\]
    If a query $\NN_H(q)$ returns $x$, then $x$ must be contained in the closed ball $\overline{B}(q, \|q\|)$. Otherwise, the distance from $q$ to $x$ would be greater than the distance $\|q\|$ from $q$ to the origin. In particular,
    \[x \in \mathbb{S}^{d-1} \cap B(q, \|q\|).\]
    This set $\mathbb{S}^{d-1} \cap B(q, \|q\|)$ is contained in the spherical cap $A_q$ where
    \[ A_q := \left\{y \in \mathbb{S}^{d-1} : \frac{q^\top y}{\|q\|} \geq 1/2\right\}.\]

    The surface area of this spherical cap is a standard computation (e.g. Section 1.2 of \cite{jenssen2018kissing}). We provide a version here with explicit constants for which we give a proof in \Cref{sec:spherical-cap}.
    
    \begin{restatable}[Surface area of a spherical cap]{lemma}{sphericalcap}
    \label{lem:spherical-cap}
    Fix $0 < \tau \leq 1$. For every integer $d \geq 2$, let
    $\mathbb{S}^{d-1} \subset \mathbb{R}^d$ be the unit sphere and let
    \[
        A = \{x \in \mathbb{S}^{d-1} : x_1 \geq \tau\}.
    \]
    Then there exists a constant $c_\tau > 0$, depending only on $\tau$,
    such that
    \[
        \frac{\vol_{d-1}(A)}
             {\vol_{d-1}(\mathbb{S}^{d-1})}
        \leq
        c_\tau \sqrt{\frac{1}{d}}
        \left(1-\tau^2\right)^{d/2}.
    \]
    \end{restatable}

    Let $q_1,\ldots, q_M$ be the (possibly random) sequence of queries that the algorithm would make if the nearest-neighbor oracle returns the origin for each query. This occurs in the event that
    \[E = \left\{x \notin \bigcup_{m=1}^M A_{q_m}\right\}.\]
    Combining the fact that $x$ is chosen uniformly from the sphere with \Cref{lem:spherical-cap}, we obtain that the probability that this event occurs is lower bounded by
    \[\Pr\left(E\right) \geq 1 - \sum_{m=1}^M \Pr\big(x \in A_{q_m}\big)  \geq 1 -  M c \left(\frac{3}{4}\right)^{d/2},\]
    where $c > 0$ is a universal constant. Notice that if the right-hand side is strictly positive, the union $\bigcup A_{q_m}$ leaves a non-negligible part of $\mathbb{S}^{d-1}$ uncovered. Thus, on the event $E$, there is almost surely no way for the learner to guess $x$ correctly. In particular, we have
    \[\Pr\big(\textrm{learner recovers $x$}\big) \leq M c \left(\frac{3}{4}\right)^{d/2}.\]
    To recover $x$ with success probability $\Omega(1)$ requires a query complexity $M = \Omega((4/3)^{d/2})$.
\end{proof}

\subsection{Multiplicative Lower Bound in the Cone} \label{sec:LB-cone}

In the case that there are $n$ hidden points, \Cref{thm:lower-bound-ball} implies an overall additive lower bound of $\Omega(n + 2^{0.2075d})$. The next lower bound gives a construction where a multiplicative $n \cdot \exp(\Omega(d))$ queries are required, showing that there are domains that admit a `direct-sum' structure, where we can embed $\Theta(n)$ independent copies of a hard construction into the space. We give a construction on the high-dimensional cone; the core source of hardness is essentially the same as for the ball.

\begin{remark} Perhaps a more immediate approach to proving a multiplicative $n \cdot \exp(\Omega(d))$ lower bound is to consider a high-dimensional cylinder in which $n$ copies of a $(d-1)$-dimensional hard instance from \Cref{thm:lower-bound-ball} are embedded into cross-sections. In order to make these instances completely independent requires sufficiently large distance between each consecutive cross-section, requiring the domain itself to grow with $n$. Thus, such a construction shows that for every $n$, there exists a domain $K_n$ where $n \cdot \exp(\Omega(d))$ queries are needed. We believe that a much more interesting statement is to show that a single fixed domain achieves the lower bound for all $n$ simultaneously. The following \Cref{thm:lower-bound-cone} achieves this using the cone. \end{remark}


\begin{theorem}[Lower bound in the cone] \label{thm:lower-bound-cone}
    For every $d \geq 2$ and height $0 < h \leq \infty$, define the cone
    \[C_d^h := \big\{(s,v) \in [0,h] \times \mathbb{R}^{d-1} : \|v\|_2 \leq s\big\}.\]
    There are constants $c_1, c_2 > 0$ such that the following holds. Let $n \geq 4$. Every randomized adaptive algorithm that exactly recovers every $n$-point hidden set from a cone $C_d^h$ with probability at least $1/10$ must make at least 
    \[c_1 n \exp(c_2 d) \quad\textrm{queries}.\]
\end{theorem}

\begin{proof} We first modify the construction in \Cref{thm:lower-bound-ball}. Let $\gamma := 1 + \sqrt{2}$. Define the truncated cone
\[T_d := \left\{(s,v) \in \left[1 - \frac{1}{\gamma}, 1 + \gamma \right] \times \mathbb{R}^{d-1} : \|v\|_2 \leq s\right\}.\]

    \noindent\begin{minipage}[t]{0.64\textwidth}
    The figure to the right visualizes $T_2 \subset \mathbb{R}^2$, which has been formed by cutting off the base and tip of a larger cone, shaded in gray. We also define four points. Let $x_L$ and $x_R$ be
    \[x_L = (1 - \gamma^{-1}, 0,\ldots, 0) \qquad\textrm{and}\qquad x_R = (1 + \gamma, 0,\ldots, 0).\]
    These are visualized by the left- and rightmost points in the figure. We also let $x_O$ and $x_U$ be points of the form
    \[x_O = (1,0,\ldots, 0)\qquad \textrm{and}\qquad x_U \in \{1\} \times \mathbb{S}^{d-2}.\]
    Choose $x_U$ uniformly at random, and let the hidden set be
    \[H_{T_d} = \{x_L, x_R, x_O, x_U\}.\]
    Notice that by our choice of $\gamma$, if a query $q \in C_d^\infty \setminus T_d$ comes from the gray region, then $x_U$ is never a nearest neighbor to $q$,
    \[x_U \ne \NN_{H_{T_d}}(q).\]
    This is because one of $x_L$ or $x_R$ is always nearer.
    \end{minipage}\hfill\begin{minipage}[t]{0.35\textwidth}
        \vspace{-15pt}
        \centering
            \begin{tikzpicture}[
    scale=1,
    line cap=round,
    line join=round,
    point/.style={
        circle,
        fill=black,
        inner sep=1.6pt
    },
    single tick/.style={
        postaction={decorate},
        decoration={
            markings,
            mark=at position 0.5 with {
                \draw[line width=0.8pt]
                    (0,-3pt) -- (0,3pt);
            }
        }
    },
    double tick/.style={
        postaction={decorate},
        decoration={
            markings,
            mark=at position 0.5 with {
                \draw[line width=0.8pt]
                    (-2pt,-3pt) -- (-2pt,3pt);
                \draw[line width=0.8pt]
                    ( 2pt,-3pt) -- ( 2pt,3pt);
            }
        }
    }
]

\pgfmathsetmacro{\a}{1 - 1/(1 + sqrt(2)) + 0.1}
\pgfmathsetmacro{\b}{2 + sqrt(2)}
\pgfmathsetmacro{\c}{0.7}

\coordinate (TL) at (\a, \a);
\coordinate (P)  at (1.2,1.2);
\coordinate (TR) at (\b, \b);
\coordinate (BR) at (\b,-\b);
\coordinate (BL) at (\a,-\a);

\coordinate (LM) at (\a,0);
\coordinate (C)  at (1.2,0);
\coordinate (RM) at (\b,0);

\coordinate (O) at (0,0);
\coordinate (TO) at (\b+\c, \b+\c);
\coordinate (BO) at (\b+\c, -\b-\c);


\fill[gray!20] (O) -- (TL) -- (BL);
\draw[thick,black!20] (TL) -- (O) -- (BL) ;

\fill[gray!20] (TO) -- (BO) -- (BR) -- (TR);
\draw[thick,black!20] (TR) -- (TO) -- (BO) -- (BR);

\draw[densely dotted] (LM) -- (RM);
\draw[densely dotted] (1.2,-1.2) -- (P);

\draw[thick] (BL) -- (BR);

\draw[thick]             (BL) -- (LM);
\draw[thick,single tick] (LM) -- (TL);

\draw[thick,single tick] (TL) -- (P);
\draw[thick,double tick] (P)  -- (TR);

\draw[thick,double tick] (RM) -- (TR);
\draw[thick]             (BR) -- (RM);

\node[point] at (LM) {};
\node[point] at (C)  {};
\node[point] at (RM) {};
\node[point] at (P)  {};

\node[below right] at (C) {$x_O$};
\node[below right] at (P) {$x_U$};
\node[left] at (LM) {$x_L$};
\node[right] at (RM) {$x_R$};
\end{tikzpicture}
    \end{minipage}
    
    \begin{claim} \label{clm:coneLBhelper}There are constants $c_1, c_2 > 0$ such that any randomized adaptive algorithm that finds $x_U$ using nearest neighbor queries in $T_d$ with probability at least $1/10$ needs $c_1 \exp (c_2 d)$ queries. \end{claim}
    
    \par\medskip Assume \Cref{clm:coneLBhelper} for now. We can deduce the lower bound for the cone $C_d^h$. In particular, we can construct $k = \lfloor n /4 \rfloor$ disjoint and independent (scaled) copies of $T_d$. Formally, let $x_U^i$ be independent draws from $\{1\} \times \mathbb{S}^{d-2}$ for $i \in [k]$. Let $\alpha_1,\ldots, \alpha_K \in (0,1)$ be scale parameters satisfying for all $i =1,\ldots, k-1$,
    \[0 < \alpha_i\left(1 - \frac{1}{\gamma}\right) < \alpha_i \left(1 + \gamma\right) < \alpha_{i+1} \left( 1 - \frac{1}{\gamma}\right) < \alpha_{i+1} ( 1 + \gamma) < h.\]
    Let the hidden set be constructed as follows:
    \[H = \bigcup_{i=1}^k \big\{\alpha_i x_L, \alpha_i x_R, \alpha_i x_O, \alpha_i x_U^i\big\}.\]
    By \Cref{clm:coneLBhelper}, finding $\alpha_i x_U^i$ requires $c_1 \exp (c_2 d)$ nearest neighbor queries from within the scaled truncated cone $\alpha_i T_d$. Since each of these truncated cones is disjoint, queries cannot be shared, and recovering every $\alpha_ix_U^i$ requires $k c_1 \exp (c_2 d)$ queries. Letting $c_1 \leftarrow c_1/4$ shows the lower bound. This completes the proof of \Cref{thm:lower-bound-cone}. \end{proof}

    \begin{proofof}{\Cref{clm:coneLBhelper}}
    We follow the same technique used in \Cref{thm:lower-bound-ball}. Let $x_U = 1 \oplus u$, where $u \in \mathbb{S}^{d-2}$. The first coordinates of $x_U$ and $x_O$ coincide, so whether a point is closer to one or the other only depends on the last $d-1$ coordinates. They are the points in the set
    \[\left\{(s,v) \in T_d : v^\top u \geq \frac{1}{2}\right\}.\]
    Since $\|v\|_2 \leq s \leq 1 + \gamma$, any point $(s,v)$ in this set satisfies
    \[\frac{v^\top u}{\|v\|} \geq \frac{1}{2(1 + \gamma)}.\]
    Thus, the unit vector $v/\|v\|$ must be chosen from a spherical cap with threshold $\tau = \frac{1}{2(1 + \gamma)}$. 
    
    Because $u \in \mathbb{S}^{d-2}$ is chosen uniformly at random, \Cref{lem:spherical-cap} shows that this event has bounded probability. In particular, there is a dimension-independent constant $c_\tau > 0$ such that
    \[\Pr\big(x_U = \NN_H(q)\big) \leq c_\tau \left(1 - \frac{1}{4 (1 + \gamma)^2}\right)^{(d-1)/2}.\]
    As in the proof of \Cref{thm:lower-bound-ball}, an algorithm that finds $x_U$ with probability at least $1/10$ makes at least $M$ queries where
    \[1 - M c_\tau\left(1 - \frac{1}{4 (1 + \gamma)^2}\right)^{(d-1)/2} < \frac{9}{10}.\]
    The claim follows by letting $c_0^{-1} = 1 - \frac{1}{4(1 + \gamma)^2}$, $c_1 = \frac{c_0 c_\tau}{10}$ and $c_2 = \frac{1}{2}\lg c_0 > 0.0156$.
    \end{proofof}

\section{Dimension Reduction in the Sphere} \label{sec:dim-reduction}

In this section we prove our dimension reduction \Cref{lem:dim-reduction} for the sphere under $\ell_2$-distance. The algorithm of \Cref{cor:sphere-alg} follows immediately by combining this lemma with \Alg{Voronoi}. We denote the unit sphere by $\Sph^{d-1}=\{x\in\R^d:\|x\|_2=1\}$. Observe that for $q,h\in\Sph^{d-1}$, we have
\[
    \|q-h\|_2^2=2-2\langle q,h\rangle.
\]
As a result, given query $q \in \Sph^{d-1}$, the $\ell_2$-nearest-neighbor oracle returns some $h \in H \subset \Sph^{d-1}$ with maximum inner product $\langle q,h\rangle$.

To prove \Cref{lem:dim-reduction}, we analyze the following iterative process. We maintain a set $X \subseteq H$ of points discovered so far and at each iteration make two queries on $\bm{w}$ and $-\bm{w}$ where $\bm{w}$ is drawn uniformly from the orthogonal subspace of the affine hull of $X$, denoted $\aff(X)$. We argue that as long as $H \not\subseteq \aff(X)$, then one of these queries will discover a new point with probability $1$.

Suppose that $X \subseteq H$ is a nonempty set of discovered points. Fixing an arbitrary $x_0 \in X$, define the linear subspace $L=\mathsf{span}(\{x-x_0:x\in X\})$ and the affine subspace $A=\aff(X)=x_0+L$. Note that $X \subset A$. Let $N=L^\perp$ denote the orthogonal subspace of $L$. In particular, note that the following holds:
\begin{align} \label{eq:x0}
    \forall w \in N, x \in A\colon ~ \langle w,x\rangle=\langle w,x_0\rangle \text{.}
\end{align}
The following lemma establishes the correctness of one iteration of the dimension reduction.

\begin{lemma} \label{lem:random-test}
Assume $N\ne\{0\}$ (equivalently, $\mathrm{dim}(L) < d$) and let $\bm{w}$  be drawn uniformly from $\Sph(N)=\{w\in N:\|w\|_2=1\}$. If $H\not\subseteq A$, then
\[
    \Pr_{\bm{w}}\left[
        \NN_H(\bm{w}) \not\in A
        \ \text{or}\
        \NN_H(-\bm{w})\not\in A
    \right]=1.
\]
In particular, querying both $\bm{w}$ and $-\bm{w}$ ensures that we learn some new point $z \in H \setminus X$ with probability $1$.
\end{lemma}

\begin{proof}
Fix $z\in H\setminus A$ and let $p=\proj_N(z-x_0)$. Since $A=x_0+L$, the assumption that $z\notin A$ implies $p\ne0$.
Therefore, the set
\[
    E_p=\{w\in\Sph(N):\langle w,p\rangle=0\}
\]
has measure zero under the uniform distribution over $\Sph(N)$.  (Note that in the
one-dimensional case, $\dim(N)=1$, we have that $E_p$ is empty.)

Now fix $w\in\Sph(N)\setminus E_p$, and choose
$\sigma\in\{-1,+1\}$ so that
$\sigma\langle w,p\rangle>0$.  For every $x\in A$, using $w\perp L$ and \Cref{eq:x0}, we have
\begin{align*}
    \langle \sigma w,z\rangle-\langle \sigma w,x\rangle =
    \sigma\langle w,z-x_0\rangle =
    \sigma\langle w,\proj_N(z-x_0)\rangle =
    \sigma\langle w,p\rangle > 0.
\end{align*}
Therefore, the query $\sigma w$ is strictly closer to $z$ than to all $x \in H \cap A$ and so the oracle must return a point from $H \setminus A$.  In particular, at least one of the two queries $\bm{w},-\bm{w}$ returns a point
outside $A$, except if $\bm{w} \in E_p$, which occurs with probability zero.
\end{proof}

\noindent \begin{proofof}{\Cref{lem:dim-reduction}} We now use \Cref{lem:random-test} to prove \Cref{lem:dim-reduction}. We will argue that the following process uses $2k+3$ queries and with probability $1$, returns a basis $u_1,\ldots,u_{k+1} \in H$ for the affine hull, $\aff(H)$. We repeat the following. Initialize the basis $X \gets \{x_0\}$ where $x_0$ is the oracle response on an arbitrarily chosen initial query in $\Sph^{d-1}$.

\begin{enumerate}
    \item Let $L=\mathsf{span}(\{x-x_0:x\in X\})$, $A=\aff(X)=x_0+L$, and $N=L^\perp$.
    \item If $N=\{0\}$, halt. Otherwise, draw $\bm{w}$ uniformly from $\Sph(N)$, query both $\bm{w}$ and $-\bm{w}$, and add both returned points to $X$.
    \item If both returned points belong to $A$, halt. Otherwise, return to step (1).
\end{enumerate}

If step (3) does not terminate, at least one newly returned point lies outside the old affine hull.  The affine dimension therefore increases by at least one.  Since $X\subseteq H$ and $H$ has rank $k$, this can happen at most $k$ times.  There is at most one final
iteration in which both answers lie in the old affine hull.  Including the initial query, this process therefore uses at most $2k+3$ queries. \end{proofof}

\section*{Acknowledgments}

\paragraph{AI Acknowledgment.} Some proof ideas in this paper were discovered through interactions with ChatGPT-5.6-Sol-Pro: the simple proof of \Cref{lem:tangent-cover}, the idea of using the cone in \Cref{thm:lower-bound-cone}, and the randomized dimension reduction over the sphere in \Cref{lem:dim-reduction}. All proofs and exposition were written by the authors.

\bibliographystyle{alpha}
\bibliography{biblio}

\appendix

\section{Bounding Spherical Cap Surface Area} \label{sec:spherical-cap}

In this section we give a complete proof of \Cref{lem:spherical-cap}, restated here for convenience.

\sphericalcap*

\begin{proof} If $\tau=1$, then $A=\{e_1\}$ has zero $(d-1)$-dimensional volume, so the result is immediate. We may assume that $0<\tau<1$. The surface volume of the unit sphere $\mathbb{S}^{d-1}$ is given by \cite[
  \href{https://dlmf.nist.gov/5.19.E4}{Eq.~(5.19.4)}]{NIST:DLMF}
\[\vol_{d-1}(\mathbb{S}^{d-1}) = \frac{2\pi^{d/2}}{\Gamma(d/2)},\]
where $\Gamma$ is the gamma function. For the spherical cap, we parametrize it using polyspherical coordinates 
\[A = \big\{(\cos \phi, \sin \phi \cdot u) : \tau \leq \cos \phi \leq 1 \textrm{ and } u \in \mathbb{S}^{d-2}\big\}.\]
The $(d-1)$-dimensional volume of the spherical cap is
\begin{align*}
    \vol_{d-1}(A) &= \int_0^{\cos^{-1} \tau} \vol_{d-2}\big(\sin \phi \cdot \mathbb{S}^{d-2}\big)\, d\phi
    \\&= \vol_{d-2}(\mathbb{S}^{d-2}) \int_0^{\cos^{-1} \tau} \sin^{d-2} \phi \, d\phi
    \\&= \frac{2 \pi^{(d-1)/2}}{\Gamma((d-1)/2)} \int_\tau^{1\vphantom{\cos^{-1}}} (1 - t^2)^{(d-3)/2}\, dt,
\end{align*}
where in the last line, we used a change of coordinates $t = \cos \phi$ with $dt = -\sin \phi \, d\phi$, along with the fact that $\sin \phi = (1 - t^2)^{1/2}$. Since $t\geq\tau$ on the interval of integration,
\begin{align*}
    \int_\tau^1 (1-t^2)^{(d-3)/2}\,dt
    &\leq
    \frac{1}{\tau}
    \int_\tau^1 t(1-t^2)^{(d-3)/2}\,dt \\
    &=
    \frac{1}{\tau(d-1)}
    (1-\tau^2)^{(d-1)/2}.
\end{align*}
The gamma function is log-convex, which implies that
\[\Gamma(d/2)^2 \leq \Gamma((d-1)/2)\Gamma((d+1)/2) = \frac{d-1}{2}\Gamma((d-1)/2)^2,\]
and hence
\[\frac{\Gamma(d/2)}{\Gamma((d-1)/2)} \leq \sqrt{\frac{d-1}{2}}.\]
Combining these inequalities yields
\[
    \frac{\vol_{d-1}(A)}
         {\vol_{d-1}(\mathbb{S}^{d-1})}
    \leq 
    \frac{\Gamma(d/2)}{\Gamma((d-1)/2)} \frac{(1 - \tau^2)^{(d-1)/2}}{\tau (d - 1)\sqrt{\pi}}
    \leq
    \frac{1}{\tau\sqrt{2\pi(d-1)}}
    (1-\tau^2)^{(d-1)/2}.
\]
Because $2(d-1)\geq d$ for $d\geq2$, this is at most
\[
    \frac{1}{\tau\sqrt{\pi(1-\tau^2)}}
    \sqrt{\frac{1}{d}}
    (1-\tau^2)^{d/2}.
\]
Thus, for $0<\tau<1$, one may take
\[
    c_\tau
    =
    \frac{1}{\tau\sqrt{\pi(1-\tau^2)}},
\]
and this completes the proof. \end{proof}

\end{document}